%% file: main.tex
\documentclass[runningheads, envcountsame, a4paper]{llncs}
\usepackage{microtype}
\usepackage[polish,UKenglish]{babel}
\usepackage{wrapfig}
\usepackage{fancyhdr}

\input{macros}

\begin{document}

\title{Weak and Nested Class Memory Automata}
\titlerunning{Weak and Nested Class Memory Automata}

\toctitle{Weak and Nested Class Memory Automata}
\tocauthor{Conrad~Cotton-Barratt, Andrzej~Murawski, and Luke~Ong}

\author{Conrad Cotton-Barratt\inst{1}\fnmsep\thanks{Supported by an EPSRC Doctoral Training Grant}, Andrzej S. Murawski\inst{2}\fnmsep\thanks{Supported by EPSRC (EP/J019577/1)}, \and C.-H. Luke Ong\inst{1},\fnmsep\thanks{Partially supported by Merton College Research Fund}}
\institute{Department of Computer Science, University of Oxford, UK 
\email{$\{$conrad.cotton-barratt,luke.ong$\}$@cs.ox.ac.uk}
\and Department of Computer Science, University of Warwick, UK
\email{a.murawski@warwick.ac.uk}}
\authorrunning{C. Cotton-Barratt, A. S. Murawski, and C.\,-H.\,L. Ong}

\maketitle
\setcounter{footnote}{0}

\thispagestyle{fancy}

\begin{abstract}
\input{abstract}
\end{abstract}

\section{Introduction}\label{sec:introduction}
\input{introduction}

\section{Preliminaries}
\input{preliminaries}

\section{Weak Class Memory Automata}\label{sec:wcma}
\input{weakcma}

\section{Nested Data Class Memory Automata}\label{sec:ndcma}
\input{nestedcma}

\input{conc}

\bibliography{bibliography}

\selectlanguage{UKenglish}

\newpage

\appendix
\section{Proof that WCMA and CCA are equivalent}\label{appendix:wcma-cca}
\input{appendix-wcma-cca}

\section{Proof that NDCMA reduce to a WSTS}\label{appendix-ndcma-wsts}
\input{appendix-ndcma-wsts}

\section{Proof that HOMCA are equivalent to HOMCA'}\label{appendix-homca}
\input{appendix-homca-proof}

\section{Proof that NDCMA are equivalent to HOMCA}\label{appendix:ndcma-homca}
\input{appendix-ndcma-homca}

\end{document}

%% file: macros.tex
\usepackage{amsmath,amsfonts}
\DeclareMathSymbol{\vDash}        {\mathrel}{AMSa}{"0F}
\DeclareMathSymbol{\leqslant}     {\mathrel}{AMSa}{"36}
\usepackage{stmaryrd}
\usepackage{color}

\usepackage[all]{xy}
\usepackage{listings}
\usepackage{enumerate}
\usepackage{paralist}

\newcommand\anglebra[1]{\left< \,#1 \, \right>}

\newcommand\lang[1]{\mathcal{L}(#1)}

\newcommand\twovec[2]{\Bigl(\negthinspace\begin{smallmatrix}#1\\#2\end{smallmatrix}\Bigr)}
\newcommand\threevec[3]{\Bigl(\negthinspace\begin{smallmatrix}#1\\#2\\#3\end{smallmatrix}\Bigr)}

\newcommand\pred{\mathit{pred}}
\newcommand\set[1]{\{#1\}}
\newcommand\powerset{{\cal P}}

\newcommand\calA{{\cal A}}
\newcommand\calB{{\cal B}}
\newcommand\calC{{\cal C}}
\newcommand\calD{{\cal D}}
\newcommand\dataset{\calD}

\newcommand\natnum{\mathbb{N}}

\def\expspace{\textsc{ExpSpace}}

\def\ptime{\textsc{PTime}}

\def\D{\mathbb{D}}

\usepackage{tikz}
\usetikzlibrary{matrix}
\usetikzlibrary{chains}
\usetikzlibrary{arrows}
\usetikzlibrary{shapes}
\usetikzlibrary{positioning}
\usetikzlibrary{calc}
\usetikzlibrary{fit}
\usetikzlibrary{automata}
\usetikzlibrary{petri}

\tikzstyle{clear} = [draw=none,fill=none]

%% file: abstract.tex
Automata over infinite alphabets have recently come to be studied extensively as potentially useful tools for solving problems in verification and database theory.  One popular model of automata studied is the Class Memory Automata (CMA), for which the emptiness problem is equivalent to Petri Net Reachability.  
We identify a restriction -- which we call weakness -- of CMA, and show that
they are equivalent to three existing forms of automata over data languages.
Further, we show that in the deterministic case they are closed under all Boolean operations, and hence have an \expspace-complete equivalence problem.
We also extend CMA to operate over multiple levels of nested data values, and show that while these have undecidable emptiness in general, adding the weakness constraint recovers decidability of emptiness, via reduction to coverability in well-structured transition systems.  We also examine connections with existing automata over nested data.

%% file: introduction.tex

A data word is a word over a finite alphabet in which every position in the word also has an associated \emph{data value}, from an infinite domain.  Data languages provide a useful formalism both for problems in database theory and verification \cite{NSV04,Segoufin06,bjorklund10}.  For example, data words can be used to model a system of a potentially unbounded number of concurrent processes: the data values are used as identifiers for the processes, and the data word then gives an interleaving of the actions of the processes.  Having expressive, decidable logics and automata over data languages then allows properties of the modelled system to be checked.  

Class memory automata (CMA) \cite{bjorklund10} are a natural form of automata over data languages.  CMA can be thought of as finite state machines extended with the ability, on reading a data value, to remember what state the automaton was in when it last saw that data value.  A run of a CMA is accepting if the following two conditions hold:
\begin{inparaenum}[(i)]
\item the run ends in a \emph{globally accepting} state; and
\item each data value read in the run was last seen in a \emph{locally accepting} state.
\end{inparaenum}
If using data values to distinguish semi-autonomous parts of a system, while the first condition can check the system as a whole has behaved correctly, the second of these conditions can be used to check that each part of the system independently behaved correctly.  The emptiness problem for class memory automata is equivalent to Petri net reachability, and while closed under intersection, union, and concatenation, they are not closed under complementation, and do not have a decidable equivalence problem.  

We earlier described how data words can be used to model concurrent systems: each process can be identified by a data value, and class memory automata can then verify properties of the system.  
What happens when these processes can spawn subprocesses, which themselves can spawn subprocesses, and so on?  In these situations the parent-child relationship between processes becomes important, and a single layer of data values cannot capture this; instead we want a notion of \emph{nested} data values, which themselves contain the parent-child relationship.  
In fact, such nested data values have applications beyond just in concurrent systems: they are prime candidates for modelling many computational situations in which names are used hierarchically.  This includes  higher-order computation where intermediate functional values are being created and named, and  later used by referring to these names. More generally, this feature is characteristic of numerous encodings into the $\pi$-calculus~\cite{San93}.

This paper is concerned with finding useful automata models which are expressive enough to decide properties we may wish to verify, as well as having good closure and decidability properties, which make them easy to abstract our queries to. 
We study a restriction of class memory automata, which we find leads to improved complexity and closure results, at the expense of expressivity.  We then extend class memory automata to a nested data setting, and find a decidable class of automata in this setting.

\bigskip

\noindent
\emph{Contributions.}
In Section \ref{sec:wcma} we identify a natural restriction of Class Memory Automata, which we call \emph{weak} Class Memory Automata, in which the local-acceptance condition of CMA is dropped.  
We show that these weak CMA are equivalent to: 
\begin{inparaenum}[(i)]
\item Class Counting Automata, which were introduced in~\cite{ManuelR09}; 
\item non-reset History Register Automata, introduced in~\cite{TzevelekosG13}; and
\item locally prefix-closed Data Automata, introduced in~\cite{Decker14}.
\end{inparaenum}

These automata have an \expspace-complete emptiness problem.  The primary advantage of having this equivalent model as a kind of Class Memory Automaton is that there is a natural notion of determinism, and we show that Deterministic Weak CMA are closed under all Boolean operations (and hence have decidable containment and equivalence problems).  

In Section \ref{sec:ndcma} we extend CMA to multiple levels of ``nested data''.  This extension is Turing-powerful in general, but reintroducing the Weakness constraint recovers decidability.  We show how these Nested Data CMA recognise the same string languages as Higher-Order Multicounter Automata, introduced in \cite{BjorklundB07}, and also how the weakness constraint corresponds to a natural weakness constraint on these Higher-Order Multicounter Automata. Finally, we show these automata to be equivalent to the Nested Data Automata introduced in \cite{Decker14}.

\bigskip

\noindent
\emph{Related Work.}
Class memory automata are equivalent to data automata (introduced in \cite{BojanczykMSSD06}), though unlike data automata, they admit a notion of determinism.  Data automata (and hence class memory automata) were shown in \cite{BojanczykMSSD06} to be equiexpressive with the two-variable fragment of existential monadic second order logic over data words.   Temporal logics have also been studied over data words \cite{DemriL09}, and the introduction of locally prefix-closed data automata and of nested data automata in \cite{Decker14} is motivated by extensions to BD-LTL, a form of LTL over multiple data values introduced in \cite{KaraSZ10}.

Fresh register automata \cite{Tzevelekos11} are a precursor to the History Register automata \cite{TzevelekosG13} which we examine a restriction of in this paper.
Class counting automata, which we show to be equivalent to weak CMA in this paper, have been extended to be equiexpressive with CMA by adding resets and counter acceptance conditions \cite{ManuelR09,Manuel11}.


We note that our restriction of class memory automata, which we call weak class memory automata, sound similar to the weak data automata introduced in \cite{KaraST12}.  However, these are two quite different restrictions, with emptiness problems of different complexities, and the two automata models should not be confused.

In the second part of this paper we examine automata over nested data values.  First-order logic over nested data values has been studied in \cite{bjorklund10}, where it was shown that the $<$ predicate quickly led to undecidability, but that only having the $+1$ predicate preserved decidability.  They also examined the link between nested data and shuffle expressions.  
In \cite{Decker14} Decker et al. introduced ND-LTL, extending BD-LTL to nested data values.  To show decidability of certain fragments of ND-LTL they extended data automata to run over nested data values, giving the nested data automata we examine in this paper. 
We note that the nested systems we introduce and study in this paper are all encodable in nested Petri nets \cite{LomazovaS99}.


%% file: preliminaries.tex

Let $\Sigma$ be a finite alphabet, and $\dataset$ be an infinite set of data values.  A \emph{data alphabet}, $\D$, is of the form $\Sigma \times \dataset$.  The set of finite data words over $\D$ is denoted $\D^*$.
The string-projection of a word in $\D^*$ is the projection to its $\Sigma$-values. We write this function $str()$, and extend it to languages over data words in the natural way.
In what follows we define the automata models that will be discussed in the paper.


\textbf{Class Memory Automata and Data Automata}.
Given a set $S$, we write $S_\bot$ to mean $S \cup \set{\bot}$, where $\bot$ is a distinguished symbol (representing a fresh data value). 
A \emph{Class Memory Automaton} \cite{bjorklund10} is a tuple $\anglebra{Q, \Sigma, q_I, \delta, F_L, F_G}$ where $Q$ is a finite set of states, $\Sigma$ is a finite alphabet, $q_I \in Q$ is the initial state, $F_G \subseteq F_L \subseteq Q$ are sets of globally- and locally-accepting sets (respectively), and $\delta$ is the transition map
$\delta : Q \times \Sigma \times Q_\bot \to \powerset(Q)$.  The automaton is deterministic if each set in the image of the transition function is a singleton.
A \emph{class memory function} is a map $f : \dataset \to Q_\bot$ such that $f(d) \neq \bot$ for only finitely many $d \in \dataset$. We view $f$ as a record of the history of computation: it holds the state of the automaton after the data value $d$ was last read, where $f(d) = \bot$ means that $d$ is fresh. A configuration of the automaton is a pair $(q, f)$ where $q \in Q$ and $f$ is a class memory function. The initial configuration is $(q_0, f_0)$ where $f_0(d) = \bot$ for every $d \in \dataset$. Suppose $(a, d) \in \Sigma \times \dataset$ is the input. The automaton can transition from configuration $(q, f)$ to configuration $(q', f')$ just if $q' \in \delta(q, a, f(d))$ and $f' = f[d \mapsto q']$. A data word $w$ is \emph{accepted} by the automaton just if the automaton can make a sequence of transitions from the initial configuration to a configuration $(q,f)$ where $q \in F_G$ and $f(d) \in F_L \cup \set{\bot}$ for every data value $d$.

A \emph{Data Automaton} \cite{BojanczykMSSD06} is a pair $(\calA, \calB)$ where $\calA$ is a letter-to-letter string transducer with output alphabet $\Gamma$, called the Base Automaton, and $\calB$ is a NFA with input alphabet $\Gamma$, called the Class Automaton.  A data word $w = w_1 \dots w_n \in \D^*$ is accepted by the automaton if there is a run of $\calA$ on the string-projection of $w$ (to $\Sigma$) with output $b_1 \dots b_n$ such that for each maximal set of positions $\set{x_1, \dots , x_k} \subseteq \set{1, \dots ,n}$ such that $w_{x_1}, \dots, w_{x_k}$ share the same data value, the word $b_{x_1} \dots b_{x_k}$ is accepted by $\calB$.

CMA and DA are expressively equivalent, with \ptime\ translation \cite{bjorklund10}.  The emptiness problem for these automata is decidable, and equivalent to Petri Net Reachability \cite{BojanczykMSSD06}.  The class of languages recognised by CMA is closed under intersection, union, and concatenation.  It is not closed under complementation or Kleene star.  Of the above, the class of languages recognised by deterministic CMA is closed only under intersection.



\textbf {Locally Prefix-Closed Data Automata.}
A Data Automaton $\calD = (\calA, \calB)$ is locally prefix-closed (pDA) \cite{Decker14} if all states in $\calB$ are final.  
The emptiness problem for pDA is \expspace-complete \cite{Decker14}.


\textbf{Class Counting Automata.}
 A \emph{bag} over $\calD$ is a function $h : \calD \to \natnum$ such that $h(d) = 0$ for all but finitely many $d \in \calD$.  
Let $C = \set{=, \neq, <, >} \times \natnum$, which we call the set of constraints.  If $c = (\texttt{op}, e) \in C$ and $n \in \natnum$ we write $n \vDash c$ iff $n \: \texttt{op} \: e$.
A \emph{Class Counting Automaton} (CCA) \cite{ManuelR09} is a tuple $\anglebra{Q, \Sigma, \Delta, q_0, F}$ where $Q$ is a finite set of states, $\Sigma$ is a finite alphabet, $q_0$ is the initial state, $F \subseteq Q$ is the set of accepting states, and $\Delta$, the transition relation, is a finite subset of $Q \times \Sigma \times C \times \set{\uparrow^+ , \downarrow} \times \natnum \times Q$.
A configuration of a CCA, $\calC = \anglebra{Q, \Sigma, \Delta, q_0, F}$, is a pair $(q, h)$ where $q \in Q$ and $h$ is a bag.  The initial configuration is $(q_0, h_0)$ where $h_0$ is the zero function.  Given a data word $w = (a_1, d_1) (a_2 , d_2) \dots (a_n, d_n)$ a run of $w$ on $\calC$ is a sequence of configurations $(q_0, h_0) (q_1, h_1) \dots (q_n , h_n)$ such that for all $0 \leq i < n$ there is a transition $(q,a,c,\pi,m,q')$ where $q = q_i$, $q' = q_{i+1}$, $a = a_{i+1}$, $h_i (d_{i+1}) \vDash c$, and
$$h_{i+1} =  \begin{cases} h_i[d_{i+1} \mapsto h_i(d_{i+1}) + m] &\mbox{ if } \pi = \uparrow^+ \\ h_i[d_{i+1} \mapsto m] &\mbox{ if } \pi = \downarrow \end{cases}$$
The run is accepting if $q_n \in F$.
The emptiness problem for Class Counting Automata was shown to be \expspace-complete in \cite{ManuelR09}.


\textbf{Non-Reset History Register Automata.}\footnote{We provide a simplified definition to that provided in \cite{TzevelekosG13}, since we do not need to consider full History Register Automata.  In particular, due to Proposition 22 in \cite{TzevelekosG13}, we need only consider histories, and not registers.}
For a positive integer $k$ write $[k]$ for the set $\set{1, 2 \dots, k}$. Fixing a positive integer $m$, define the set of labels \textsf{Lab} $= \powerset([m])^2$. 
A non-reset History Register Automaton (nrHRA) of type $m$ with initially empty assignment is a tuple $\calA = \anglebra{Q, \Sigma, \delta, q_0, F}$ where $q_0 \in Q$ is the initial state, $F \subseteq Q$ is the set of final states, and $\delta\subseteq Q \times \Sigma \times \textsf{Lab} \times Q$.  
A configuration of $\calA$ is a pair $(q,H)$ where $q \in Q$ and $H : [m] \rightarrow \powerset_\mathit{fn}(\calD)$ where $\powerset_\mathit{fn} (\calD)$ is the set of finite subsets of $\calD$.  We call $H$ an \emph{assignment}, and for $d \in \calD$ we write $H^{-1}(d)$ for the set $\set{i \in [m] : d \in H(i)}$.  The initial configuration is $(q_0, H_0)$, where $H_0$ assigns every integer in $[m]$ to the empty set.  
When the automaton is in configuration $(q,H)$, on reading input $\twovec{a}{d}$ it can transition to configuration $(q',H')$ providing there exists $X \subseteq [m]$ such that $(q,a,(H^{-1}(d),X),d) \in \delta$ and $H'$ is obtained by removing $d$ from $H(i)$ for each $i$ then adding $d$ to each $H(i)$ such that $i \in X$.  
A run is defined in the usual way, and a run is accepting if it ends in a configuration $(q,H)$ where $q \in F$.


\textbf{Higher-Order Multicounter Automata.} 
A multiset over a set $A$ is a function $m : A \rightarrow \natnum$.  A level-1 multiset over $A$ is a finite multiset over $A$.  A level-$(k+1)$ multiset over $A$ is a finite multiset of level-$k$ multisets over $A$.
We can visualise this with nested set notation: e.g. $ \set{ \set{a,a} , \set{} , \set{}}$ represents the level-2 multiset containing one level-1 multiset containing two copies of $a$, and two empty level-1 multisets.
A multiset is \emph{hereditarily empty} if, written in nested set notation, it contains no symbols from $A$.

Higher-Order Multicounter Automata (HOMCA) were introduced in \cite{BjorklundB07}, and their emptiness problem was shown to be Turing-complete at level-2 and above. 
A level-$k$ multicounter automaton is a tuple $\anglebra{Q, \Sigma, A, \Delta, q_0, F}$ where $Q$ is a finite set of states, $\Sigma$ is the input alphabet, $A$ is the multiset alphabet, $q_0$ is the initial state, and $F$ is the set of final states.  A configuration is a tuple $(q, m_1, m_2, \dots, m_k)$ where $q \in Q$ and each $m_i$ is either undefined ($\bot$) or a level-$i$ multiset over $A$.  The initial configuration is $(q_0, \bot, \dots, \bot)$.
$\Delta$ is the transition relation, and is a subset of $Q \times \Sigma \times ops \times Q$ where $ops$ is the set of possible counter operations.  These operations, and meanings, are as follows:
\begin{inparaenum}[(i)]
\item $new_i$ ($i \leq k$) turns $m_i$ from $\bot$ into the empty level-$i$ multiset;
\item $inc_a$ ($a \in A$) adds $a$ to $m_1$;
\item $dec_a$ ($a \in A$) removes $a$ from $m_1$;
\item $store_i$ ($ i < k$) adds $m_i$ to $m_{i+1}$ and sets $m_i$ to $\bot$;
\item $load_i$ ($i < k$) non-deterministically removes an $m$ from $m_{i+1}$ and turns $m_i$ from $\bot$ to $m$.  This can happen only when $m_1 \dots m_i$ are all $\bot$.
\end{inparaenum}
The automaton reads the input word from left to right, updating $m_1 \dots m_k$ as determined by the transitions.  A word is accepted by the automaton just if there is a run of the word such that the automaton ends up in configuration $(q, m_1, \dots, m_k)$ where $q \in F$ and each $m_i$ is hereditarily empty.

%% file: weakcma.tex
In this section we introduce a restriction of class memory automata, \emph{weak} class memory automata (WCMA), and discuss the improved closure and complexity properties.  We also show that 
WCMA correspond to a natural restriction of data automata, locally-prefix closed data automata, as well as two other independent automata models, class counting automata and non-reset history register automata.

\begin{definition}
A class memory automaton $\anglebra{Q, \Sigma, \Delta, q_0, F_L, F_G}$ is \emph{weak} if all states are locally accepting (i.e. $F_L = Q$).  
\end{definition}
When defining a weak CMA (WCMA) we may omit the set of locally accepting states, and just give one set of final states, $F$.

The emptiness problem for class memory automata is reducible (in fact, equivalent) to emptiness of multicounter automata (MCA) \cite{BojanczykMSSD06,bjorklund10}.  This reduction works by using counters to store the number of data values last seen in each state.  The local-acceptance condition is checked by the zero-test of each counter at the end of 
a run of an MCA.  In the weak CMA case, this check is no longer necessary, and so emptiness is reducible  to emptiness of weak MCA.  
Just as MCA emptiness is equivalent to Petri net reachability, weak MCA emptiness is equivalent to Petri net coverability.  
\begin{example}\label{ex:wcma-pncoverability}
We give an example showing how a very simple Petri net reachability query can be reduced to an emptiness of CMA problem, and the small change required to reduce coverability queries to emptiness of WCMA.
The idea is to encode tokens in the Petri net using data values: the location of the token is stored by the class memory function's memory for the data value.  Transitions in the Petri net will be simulated by sequences of transitions in the automaton, which change class memory function appropriately.  
Consider the Petri net shown in Figure~\ref{fig:example-pn}, with initial marking on the left and target marking on the right.
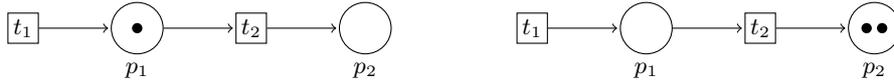
\begin{figure}[h]
   \vspace{-10pt}
\begin{tikzpicture}[node distance=1.5cm]
\node [place,tokens=1] (s1) [label=below:$p_1$]   {}   ;
\node [clear] (space) [right of=s1] {};
\node [place] (s2) [right of=space,label=below:$p_2$] {};

\node [transition] (e1) [right of=s1] {$t_2$}
      edge [pre]                  (s1)
      edge [post]                (s2);
\node [transition] (e2) [left of=s1] {$t_1$}
      edge [post]                  (s1);
            
\node [place] (p1) [right=3cm of s2,label=below:$p_1$]    {}   ;
\node [clear] (space2) [right of=p1] {};
\node [place,tokens=2] (p2) [right of=space2,label=below:$p_2$] {};

\node [transition] (f1) [right of=p1] {$t_2$}
      edge [pre]                  (p1)
      edge [post]                (p2);
\node [transition] (f2) [left of=p1] {$t_1$}
      edge [post]               (p1);
\end{tikzpicture}
   \vspace{-10pt}
\caption{An example Petri net with initial and target markings.}
\label{fig:example-pn}
   \vspace{-20pt}
\end{figure}
\end{example}
We give the automaton which models this reachability query in Figure~\ref{fig:example-cma}.  The first transitions from the initial state just set up the initial marking.  As there is only one token in the initial marking, this just involves reading one fresh data value: this is the transition from $s_0$ to $s_1$ below.  Once the initial marking has been set up (reaching $s_2$ below), the automaton can simulate the transitions firing any number of times.  Each loop from $s_2$ back to itself represents one transition in the Petri net firing: the loop above represents $t_1$ firing, and the loop below represents $t_2$ firing.  For $t_1$ to fire, no preconditions must be met, and a new data value can be read in state $s_3$, thus data values last seen in either of states $s_1$ and $s_3$ represent tokens in $p_1$.    For $t_2$ to fire, a token must be removed from $p_1$, since tokens in $p_1$ are represented by tokens in either $s_1$ or $s_3$, the first transition in this loop -- to $s_4$ -- involves reading a data value last seen in one of these states.  Thus data values seen in $s_4$ represent removed tokens, which we do not use again.  Then a new token is placed in $p_2$ by reading a fresh data value in $s_5$.  Once back in $s_2$ these loops can be taken more, or the fact that a marking covering the target marking has been reached can be checked by reading two data values last seen in $s_5$ to reach a final state.  The only globally accepting state is $s_7$, and all states except those which are used to represent tokens -- i.e. all except $s_1$, $s_3$, and $s_5$ -- are locally accepting.  The local acceptance condition thus checks that no other tokens remain in the simulated Petri net.
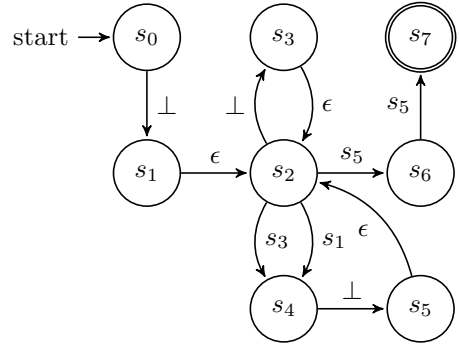
\begin{wrapfigure}{r}{0.5\textwidth}
   \vspace{-15pt}
\begin{tikzpicture}[->,>=stealth',shorten >=1pt,auto,node distance=1.8cm, semithick]]
\node [state,initial] (s0)  {$s_0$};
\node [state] (s1)  [below of=s0] {$s_1$};
\node [state] (s2)  [right of=s1] {$s_2$};
\node [state] (s4)  [below of=s2] {$s_4$};
\node [state] (s5)  [right of=s4] {$s_5$};
\node [state] (s3)  [above of=s2] {$s_3$};
\node [state] (s6)  [right of=s2] {$s_6$};
\node [state,accepting] (s7) [above of=s6] {$s_7$};

\path (s0) edge    node {$\bot$} (s1)
      (s1) edge    node {$\epsilon$} (s2)
      (s2) edge  [bend left]  node {$s_1$} (s4)
      (s2) edge  [bend right]  node {$s_3$} (s4)
      (s4) edge    node {$\bot$} (s5)
      (s5) edge  [bend right]  node {$\epsilon$} (s2)
      (s2) edge  [bend left]  node {$\bot$} (s3)
      (s3) edge  [bend left]  node {$\epsilon$} (s2)
      (s2) edge    node {$s_5$} (s6)
      (s6) edge    node {$s_5$} (s7);
\end{tikzpicture}
   \vspace{-10pt}
\caption{A class memory automaton simulating the Petri net query shown in Figure \ref{fig:example-pn}.}
\label{fig:example-cma}
   \vspace{-25pt}
\end{wrapfigure}

If we were interested in a coverability query, the same automaton, but without the local acceptance condition, would obviously suffice.  Thus emptiness of WCMA is equivalent to Petri net coverability, which is \expspace-complete.


We now give the main observation of this section: that weak CMA are equivalent to three independent existing automata models.
\begin{theorem}
Weak CMA, locally prefix-closed DA, class counting automata, and non-reset history register automata are all \ptime-equivalent. 
\end{theorem}
\begin{proof}
That Weak CMA and pDA are equivalent is a simple alteration of the proof of equivalence of CMA and DA provided in \cite{bjorklund10}.  

Recall that CCA use a ``bag'', which essentially gives a counter for each data value.  Weak CMA can easily be simulated by CCA by identifying each state with a natural number; then the bag can easily simulate the class memory function, by setting the data value's counter to the appropriate number when it is read.  To simulate a CCA with a WCMA, we first observe that for any CCA, since counter values can only be incremented or reset, there is a natural number, $N$, above which different counter values are indistinguishable to the automaton.  Thus we need only worry about a finite set of values.  This means the value for the counter of each data value can be stored in the automaton state, and thereby the class memory function.  A full proof is provided in Appendix \ref{appendix:wcma-cca}.

In \cite{TzevelekosG13} the authors already show that nrHRA can be simulated by CMA.  Their construction does not make use of the local-acceptance condition, so the fact that nrHRA can be simulated by WCMA is immediate.  In order to simulate a given WCMA with state set $[m]$, one can take a nrHRA of type $m$, 
with place $i$ storing the data values last seen in state $i$.
\end{proof}
Data automata, and hence pDA, unlike CMA and WCMA, do not have a natural notion of determinism, nor a natural restriction corresponding to deterministic CMA or WCMA.  
What about for CCA?  
We define CCA to be deterministic if for each state $q$ and input letter $a$, the transitions $(q, a, c, \dots) \in \Delta$ are such that the $c$'s partition $\natnum$.  
The translations provided in Appendix \ref{appendix:wcma-cca} also show that deterministic WCMA and deterministic CCA are equivalent.  
We can ask the same question of non-reset HRA.
We find that the natural notion of determinism here is:  for each $q \in Q, a \in \Sigma$, and $X \subseteq [m]$ there is precisely one $Y \subseteq [m]$ and $q' \in Q$ such that $(q,a,(X,Y),q') \in \delta$. Similarly, the translations discussed above show deterministic WCMA to be equivalent to deterministic nrHRA.


It follows from the results for CCA in \cite{ManuelR09} that Weak CMA, like normal CMA, are closed under intersection and union, though these closures can easily be shown directly using product constructions.  In fact, Deterministic Weak CMA have even nicer closure properties:
\begin{proposition}
Deterministic Weak CMA are closed under all Boolean operations.
\end{proposition}
\begin{proof}
Closure under intersection and union can be shown by product constructions.  For complementation one can use the same method as for DFA: complementing the final states.
\end{proof}

\begin{corollary}
The containment and equivalence problems for Deterministic Weak CMA are \expspace-complete.
\end{corollary}



%% file: nestedcma.tex

In Section \ref{sec:introduction} we discussed how data values fail to provide a good model for modelling computations in which names are used hierarchically, such as a system of concurrent processes which can spawn subprocesses.
Motivated by these applications, in this section we introduce a notion of nested data values in which the data set has a forest-structure.  This is a stylistically different presentation 
to earlier work on nested data in that~\cite{bjorklund10,Decker14} require that each position in the words considered have a data value in each of a fixed number of levels.  By giving the data set a forest-structure, we can explicitly handle variable levels of nesting within a word.  However, we note that there is a natural translation between the two presentations.


\begin{definition}
A \emph{rooted tree} (henceforth, just \emph{tree}) is a simple directed graph $\anglebra{D, \pred}$, where $\pred : D \rightharpoonup D$ is the predecessor map defined on every node of the tree except the root, such that every node has a unique path to the root. A node $n$ of a tree has \emph{level} $l$ just if $\pred^{l-1}(n)$ is the root (thus the root has level 1). A tree has \emph{bounded level} just if there exists a least $l \geq 1$ such that every node has level no more than $l$; we say that such a tree has level $l$.

We define a \emph{nested dataset} $\anglebra{\dataset, \pred}$ to be a forest of infinitely many trees of level $l$ which is \emph{full} in the sense that for each data value $d$ of level less than $l$, 
$d$ has infinitely many children.
\end{definition}


We now extend CMA to nested data by allowing the nested data class memory automaton (NDCMA), on reading a data value $d$, to access the class memory function's memory of not only $d$, but each ancestor of $d$ in the nested data set.  Once a transition has been made, the class memory function updates the remembered state not only of $d$, but also of each of its ancestors.  Formally:


\begin{definition}
Fix a nested data set of level $l$.  A \emph{Nested Data CMA of level $l$} is a tuple $\anglebra{Q,\Sigma,\delta,q_0,F_L,F_G}$ where $Q$ is a finite set of states, $q_0 \in Q$ is the initial state, $F_G \subseteq F_L \subseteq Q$ are sets of globally and locally accepting states respectively, and $\delta$ is the transition map.  $\delta$ is given by a union $\delta = \bigcup_{1 \leq i \leq l} \delta_i$ where each $\delta_i$ is a function:
\[
\delta_i : Q \times \Sigma \times (\set{i} \times (Q_\bot)^{i}) \rightarrow \powerset (Q)
\]
The automaton is \emph{deterministic} if each set in the image of $\delta$ is a singleton; and is \emph{weak} if $F_L = Q$.
A configuration is a pair $(q,f)$ where $q \in Q$, and $f: \dataset \rightarrow Q_\bot$ is a class memory function (i.e. $f(d) = \bot$ for all but finitely many $d \in \dataset$).  The initial configuration is $(q_0, f_0)$ where $f_0$ is the class memory function mapping every data value to $\bot$.  A configuration $(q,f)$ is final if $q \in F_G$ and $f(d) \in F_L \cup \set{\bot}$ for all $d \in \dataset$.
The automaton can transition from configuration $(q,f)$ to configuration $(q',f')$ on reading input $\twovec{a}{d}$ just if $d$ is a level-$i$ data value, $q' \in \delta ( q, a, (i, f(pred^{i-1} (d)), \dots, f(pred(d)), f(d)))$, and $f' = f[d \mapsto q, pred(d) \mapsto q, \dots, pred^{i-1}(d) \mapsto q]$.
A run $(q_0, f_0), (q_1, f_1), \dots, (q_n, f_n)$ is accepting if the configuration $(q_n, f_n)$ is final.  $w \in L(\calA)$ if there is an accepting run of $\calA$ on $w$.
\end{definition}

It is clear that level-1 NDCMA are equivalent to normal CMA.  
We know that emptiness of class memory automata is equivalent to reachability in Petri nets; it is natural to ask whether there is any analogous correspondence -- to some kind of high-level Petri net -- once nested data is used. 

\begin{example}\label{ex:ndcma-pnwr}
In Example \ref{ex:wcma-pncoverability}, we showed how CMA (resp. WCMA) can encode Petri net reachability (resp. coverability).  A similar technique allows reachability (resp. coverability) of Petri nets with reset arcs to be reduced to emptiness of NDCMA (resp. weak NDCMA).  
The key idea is to have, for each place in the net, a level-1 data value -- essentially as a ``bag'' holding the tokens for that place. Nested under the level-1 data value, level-2 data values are used to represent tokens just as before.  
When a reset arc is fired, the corresponding level-1 data value is moved to a ``dead'' state -- from where it and the data values nested under it are not moved again -- and a fresh level-1 data value is then used to hold subsequently added tokens to that place.
\end{example}

\begin{theorem}\label{thm:ndcma-emptiness}
The emptiness problem for NDCMA is undecidable.  Emptiness of Weak NDCMA is decidable, but Ackermann-hard.
\end{theorem}
\begin{proof}
This result follows from Theorem \ref{thm:NDCMA-pNDA} together with results in \cite{Decker14}, though we also provide a direct proof.

We show decidability by reduction to a well-structured transition system \cite{FinkelS01} constructed as follows: a class memory function on a nested data set can be viewed as a labelling of the data set by labels from the set of states.  Since we only care about the shape of the class memory function (i.e. up to automorphisms of the nested data set), we can remove the nodes labelled by $\bot$, and view a class memory function as a finite set of labelled trees.  The set of finite forests of finite trees of bounded depth with the order given by $F \leq F'$ iff there is a forest homomorphism from $F$ to $F'$ (where a forest is the natural extension of tree homomorphisms to forests) is a well-quasi-order \cite{Meyer08}, which provides the basis for the well-structured transition system.  A full proof is given in Appendix \ref{appendix-ndcma-wsts}.

Undecidability for NDCMA and Ackermann-hardness for Weak NDCMA follow from the ideas in Example \ref{ex:ndcma-pnwr}: the reachability (resp. coverability) problem for Petri nets with reset arcs is encodable in NDCMA (resp. Weak NDCMA), and this is undecidable~\cite{ArakiK76} (resp. Ackermann-hard \cite{Schnoebelen02}).
\end{proof}
Weak Nested Data CMA have similar closure properties to Weak CMA.
\begin{proposition}
\begin{enumerate}[(i)]
\item Weak NDCMA are closed under intersection and union.
\item Deterministic Weak NDCMA are closed under all Boolean operations.
\end{enumerate}
\end{proposition}
\begin{proof}
Again these can be shown by the same techniques as for DFA.
\end{proof}

\begin{corollary}
The containment and equivalence problems for
Deterministic Weak NDCMA are decidable.
\end{corollary}

\subsection{Link with Nested Data Automata}
In \cite{Decker14} Decker et al. also examined ``Nested Data Automata'' (NDA), and showed the locally prefix-closed NDA (pNDA) to have decidable emptiness (via reduction to well-structured transition systems).  In fact, these NDA precisely correspond to NDCMA, and again being locally prefix-closed corresponds to weakness. In this section we briefly outline this connection.

\textbf{Nested Data Automata.}  A $k$-nested data automaton is a tuple $(\calA, \calB_1, \calB_2, \dots, \calB_k)$ where $(\calA, \calB_i)$ is a data automaton for each $i$.  Such automata run on words over the alphabet $\Sigma \times \calD^k$, where $\calD$ is a (normal, unstructured) dataset.  As for normal data automata, the transducer, $\calA$, runs on the string projection of the word, giving output $w$.  Then for each $i$ the class automaton $\calB_i$ runs on each subsequence of $w$ corresponding to the positions which agree on the first $i$ data values.   The NDA is locally prefix-closed if each $(\calA, \calB_i)$ is.

Since these NDA are defined on a slightly different presentation of nested data, we provide the following presentation of NDCMA over multiple levels of data.
\begin{definition}
A \emph{Nested Data CMA of level $k$} over the alphabet $\Sigma \times \calD^k$ is a tuple $\anglebra{Q,\Sigma,\delta,q_0,F_L,F_G}$ where $Q$ is a finite set of states, $q_0 \in Q$, $F_G \subseteq F_L \subseteq Q$, and $\delta : Q \times \Sigma \times (Q_\bot)^{k} \rightarrow \powerset (Q)$ is the transition map.

A configuration is a tuple $(q,f_1,f_2,\dots, f_k)$, where each $f_i : \calD^i \rightarrow Q_\bot$ maps an $i$-tuple of data values to a state in the automaton (or $\bot$).  The initial configuration is $(q_0, f_1^0, \dots , f_k^0)$ where $f_i^0$ maps every tuple in the domain to $\bot$.  A configuration $(q,f_1, \dots ,f_k)$ is final if each $f_i$ maps into $F_L \cup \set{\bot}$.  The automaton can transition from configuration $(q,f_1, \dots, f_k)$ to configuration $(q', f_1', \dots, f_k')$ on reading input $(a, d_1, \dots, d_k)$ just if $q' \in \delta(q,a,(f_1(d_1), f_2(d_1,d_2), \dots , f_k (d_1, \dots, d_k)))$, and each $f_i' = f_i[(d_1, \dots, d_i) \mapsto q']$.  
\end{definition}
Using ideas from the proof of equivalence between CMA and DA in~\cite{bjorklund10}, we can show the following result:
\begin{theorem}\label{thm:NDCMA-pNDA}
NDCMA (resp. weak NDCMA) and NDA (resp. pNDA) are expressively equivalent, with effective translations.
\end{theorem}

\subsection{Link with Higher-Order Multicounter Automata}

In \cite{BjorklundB07} the authors examined a link between nested data values and shuffle expressions.  In doing so, they introduced higher-order multicounter automata (HOMCA).  While not explicitly over nested data values, they are closely related to the ideas involved, and in fact we show that, just as multicounter automata and CMA are equivalent, there is a natural translation between HOMCA and the NDCMA we have introduced.  Further, just as the equivalence between MCA and CMA descends to one between weak multicounter automata and weak CMA, we find an equivalence between weak NDCMA and ``weak'' HOMCA in which the corresponding acceptance condition -- hereditary emptiness -- is dropped.  
To show this, we introduce HOMCA', which add restrictions to the $store_i$ and $new_i$ counter operations analogous to the restriction for the $load_i$ operation.
\begin{wrapfigure}{r}{0.7\textwidth}
   \vspace{-30pt}
\begin{displaymath}
\xymatrix{
\text{HOMCA} \ar@<.5ex>[r] & \text{HOMCA'} \ar@<.5ex>[r] \ar@<.5ex>[l] & \text{NDCMA} \ar@<.5ex>[l]  \\
\ar[u] \ar@<.5ex>[r] \text{weak HOMCA} &  \ar[u] \ar@<.5ex>[r] \ar@<.5ex>[l] \text{weak HOMCA'}  & \text{weak NDCMA} \ar[u] \ar@<.5ex>[l]
}
\end{displaymath}
   \vspace{-10pt}
\caption{A diagram showing translations between HOMCA, HOMCA', NDCMA, and their weak counterparts.}
   \vspace{-10pt}
\label{fig:homca-equivalences}
\end{wrapfigure}
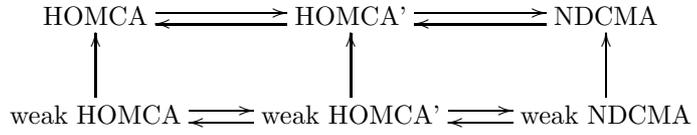
We show that these HOMCA' are equivalent to HOMCA, and that HOMCA' are equivalent to NDCMA, with both of these equivalences descending to the weak versions.  These equivalences are summarised in Figure~\ref{fig:homca-equivalences}.

\begin{definition}
We define \emph{weak} HOMCA to be just as HOMCA, but without the hereditary-emptiness condition on acceptance,  i.e. a run is accepting just if it ends in a final state.
\end{definition}

\begin{definition}
We define HOMCA' to be the same as HOMCA, except for the following changes to the counter operations:
\begin{inparaenum}[(i)]
\item $store_i$ operations are only enabled when $m_1 = m_2 = \dots = m_{i-1} = \bot$; and
\item $new_i$ operations are only enabled when $m_k \neq \bot, m_{k-1} \neq \bot, \dots, m_{i+1} \neq \bot$ and $m_{i-1} = m_{i-2} = \dots = m_1 = \bot$.
\end{inparaenum}

As for HOMCA, we define \emph{weak} HOMCA' to be HOMCA' without the hereditary-emptiness condition. 
\end{definition}
This means that each reachable configuration $(q, m_1, \dots , m_k)$ is such that there is a unique $0 \leq i \leq k$ such that for all $j \leq i$, $m_{j} = \bot$ and each $l > i$, $m_{l} \neq \bot$.

\begin{theorem}
HOMCA (resp. weak HOMCA) and HOMCA' (resp. weak HOMCA') are expressively equivalent, with effective translations between them.
\end{theorem}
\begin{proof}
This requires simulating the HOMCA operations $store_i$ and $new_i$ in HOMCA': which can be difficult if, for instance, the HOMCA is carrying out a $store_i$ operation when it has a current level-$(i-1)$ multiset in memory.  The trick is to move the level-$(i-1)$ multiset across to be nested under a new level-$i$ multiset, and this can be done one element at a time in a ``folding-and-unfolding'' method.  The hereditary emptiness condition checks that the each of these movements was completed, i.e. no element was left unmoved. In the weak case some elements not being moved could not change an accepting run to a non-accepting run, so the fallibility of the moving method does not matter.   We provide a more detailed proof sketch in Appendix~\ref{appendix-homca}.
\end{proof}

We now have the main result of this section:
\begin{theorem}\label{thm:ndcma-homca-equivalence}
For every (weak) level-$k$ NDCMA, $\calA$, there is a (weak) level-$k$ HOMCA', $\calA'$, such that $\lang{\calA'} = str(\lang{\calA})$, and vice-versa. 
\end{theorem}
\begin{proof}
This proof rests on the strong similarity between the nesting of data values, and the nesting of level-$i$ multisets in level-$(i+1)$ multisets.  

For a NDCMA to simulate a HOMCA', we use level-$k$ data values to represent instances of the multiset letters, level-$(k-1)$ data values to represent level-1 multisets, and so on, up to level-$1$ data values representing level-$(k-1)$ multisets.  Since each run of a HOMCA' can have at most one level-$k$ multiset, this does not need to be encoded in data values.  

Conversely, when simulating a NDCMA with a HOMCA', a level-$k$ data value is represented by an instance of an appropriate multiset letter.  The letter contains the information on which state the data value was last seen in.  Level-$(k-1)$ data values are represented by level-1 multisets, which also include a multiset letter storing the state that data value was last seen in.  We provide a full proof of the theorem in Appendix~\ref{appendix:ndcma-homca}.
\end{proof}

%% file: conc.tex

\section{Conclusion}

We showed that by dropping local acceptance conditions in Class Memory Automata one obtains a robust class of languages that has already appeared in the literature
under various guises. Furthermore, its deterministic restriction is closed under all Boolean operations, which implies decidability (in fact, \expspace-completeness)
of the containment and equivalence problems. We recall that both were undecidable for general class memory automata.

We introduce a new notion of nesting for data languages, based on tree-structured datasets. This notion does not commit all letters to be at the same level of nesting
and appears promising from the point of view of modelling scenarios with hierarchical name structure, such as concurrent or higher-order computation. 
We extend Class Memory Automata to operate over these nested datasets, and show that without the local acceptance condition, these have a decidable emptiness problem, and in the deterministic case are closed under all Boolean operations.

In future work, we would like to understand better whether there is a natural fragment of the $\pi$-calculus that corresponds to the new classes of automata.
On the logical side, an interesting outstanding question is to characterize languages accepted by our classes of automata
with suitable logics.

%% file: appendix-wcma-cca.tex

\begin{proposition}
Given a WCMA a CCA recognising the same language can be constructed in \ptime.
\end{proposition}
\begin{proof}
Let $\calA = \anglebra{Q, \Sigma, \Delta, q_0, F}$ be a WCMA.  WLOG suppose $Q = \set{q_1, \dots , q_n}$.  Consider the CCA $\calA' = \anglebra{Q, \Sigma, \Delta', q_0, F}$ where $\Delta'$ is given by:
\begin{align*}
\Delta' = &\set{ ( q_i,a,(=,k),\downarrow,j,q_j ) \:\: | \:\:  k \geq 1 \text{ and } (q_i, a, q_k, q_j) \in \Delta} \\ &\cup \set{ ( q_i,a,(=,0),\downarrow,j,q_j ) \:\: | \:\: (q_i, a, \bot, q_j) \in \Delta }
\end{align*}
The bag is used to simulate the class memory function by representing a state by a natural number.  That these automata recognise the same language is a straightforward induction.
\end{proof}

\begin{proposition}
Given a CCA a WCMA recognising the same language can be constructed in \ptime.
\end{proposition}
\begin{proof}
The idea of this construction is to use the class memory function to mimic the counters, which can be done by representing the counter value in the states.  This is possible since only a finite number of counter values can behave differently: if the greatest number used in the constraints of the transition function of a CCA is $n_0$ then all the integers above $n_0$ are equivalent for the purposes of which transitions apply.  Further, since the counters can only be incremented or reset, there is no need to track how much greater than $n_0$ such a counter is.

Let $\calA = \anglebra{Q, \Sigma, \Delta, q_0, F}$ be a CCA.  Since $\Delta$ is finite, there is a greatest integer used in $\Delta$.  Let $n_0$ be this greatest occurring integer.  Define $N = \set{0, 1, \dots, n_0, n_0+1}$.

We define the WCMA $\calA' = \anglebra{Q \times N, \Sigma, \Delta', (q_0, 0), F \times N}$ where $\Delta' = \bigcup_{\delta \in \Delta} S_\delta$
where, for $\delta = (q,a,c,\uparrow^+,m,q')$: 
\begin{align*}
S_\delta = &\set{((q,i), a, \bot, (q',m)) \: | \: i \in N, 0 \vDash c} \\
&\cup \set{((q,i), a, (q'',l), (q',j)) \: | \: i,l \in N, l \vDash c, j = min(l+m, n_0+1)}
\end{align*}
and for $\delta = (q,a,c,\downarrow,m,q')$: 
\begin{align*}
S_\delta = &\set{((q,i), a, \bot, (q',m)) \: | \: i \in N, 0 \vDash c} \\
&\cup \set{((q,i), a, (q'',l), (q',m)) \: | \: i,l \in N, l \vDash c}
\end{align*}
Again, showing that these automata recognise the same language is a straightforward induction.
Note that this construction is in \ptime\ since it is assumed that integers in CCAs are given in unary.
\end{proof}

It follows immediately from the corresponding result for CCA that the emptiness problem for WCMA is \expspace-complete.  We note that although here we have provided an equivalence with CCA, this is also straightforward to show by direct reduction to weak multicounter automata: the counters can be used to count how many data values were last seen in each state.\footnote{Indeed, our reasons for calling these types of CMA ``weak" stems from the equivalence between these weak CMA and weak multicounter automata, mirroring the equivalence between (strong) CMA and (strong) multicounter automata}

%% file: appendix-ndcma-wsts.tex
We assume familiarity with the theory of well-structured transition systems, as described in \cite{FinkelS01}.

For this we will fix a nested data set, $\anglebra{\calD, pred}$ of level $k$, and we make the following observation: a class memory function $f : \calD \rightarrow Q_\bot$ can be interpreted as a labelling of the forest $\calD$ with labels from $Q_\bot$.  To simply the following discussion, we add a distinguished ``level-0'' data value, which is the parent of all of the level-1 data values; this means we are working only over a tree, rather than a forest.  Further, we note that by the definition of the transition function, if $f$ is a class memory function obtained as part of a reachable configuration $(q,f)$, then $f(d) = p \in Q$ implies $f(pred(d)) \neq \bot$ (here we assume that the class memory function has some distinguished symbol for the level-0 data value). Hence, in our labelled tree $\calD_f$, whenever we reach a node labelled $\bot$, all of that node's descendants are also labelled $\bot$.  Thus, each reachable class memory function gives rise to an unordered labelled finite tree, $T_f$ (with labels from $Q$ plus a distinguished symbol for the root) of depth $\leqslant k+1$.  Further, given two class memory functions $f$ and $f'$ such that $T_f = T_{f'}$, it is clear that $f$ is equivalent to $f'$ up to automorphism of $\calD$ (i.e. renaming of data values).  

Given an NDCMA $\calA = \anglebra{Q, \Sigma, \Delta, q_0, F}$ we construct the WSTS $\mathcal{S} = \anglebra{S, \rightarrow, \leq}$ as follows:
\begin{itemize}
\item $S = Q \times \Phi$ where $\Phi$ is the set of unordered labelled finite trees of depth $\leqslant k +1$ (with labels from $Q$ plus a distinguished symbol for the root)
\item $\rightarrow$ is defined as follows: $(q, T) \rightarrow (q', T')$ iff there are Class Memory Functions $f_T$ and $f_{T'}$ s.t. $T_{f_T} = T$ and $T_{f_{T'}} = T'$ and $(q', f_{T'})$ is a valid successor configuration (in $\calA$) of $(q, f_T)$.
\item $\leq$ is defined as follows: $(q,T) \leq (q',T')$ iff $q = q'$ and there is a tree homomorphism from $T$ to $T'$.\footnote{A tree homomorphism from $T_1$ to $T_2$ is a function, $\phi$, mapping each node of $T_1$ to a node of $T_2$ which maps the root node to the root node, and preserves $pred$ and labels. (i.e. $\phi ( pred (n)) = pred ( \phi(n))$ and $lab(n) = lab (\phi(n))$.)}
\end{itemize}

That $\rightarrow$ is finite-branching is obvious.  That $\leq$ is a qo is immediate from composition of homomorphisms.  That it is a well-quasi-order is given by Lemma 7 in \cite{Meyer08}.  The upward-compatibility condition is straightforward: if $(p,R) \leq (q,T)$ and $(p,R) \rightarrow (p',R')$ then $p = q$ and there is some $\delta \in \Delta$ such that applying $\delta$ from configuration $(p,f_R)$ yields configuration $(p, f_{R'})$.  Applying the same $\delta$ from configuration $(q,T)$ will yield an appropriate $(p', T') \geq (p', R')$.

Decidability of $\leq$ is also straightforward: to decide whether $(q, T) \leq (q', T')$ first check $q = q'$, then simply check each possible injection from $T$ to $T'$ for being a homomorphism.

The only remaining property to check is that we have an effective pred-basis.  We do this by showing an effective pred-basis wrt each transition in $\calA$.  The union of these can then be taken.  Formally, we compute $pb(s)$ as $pb(s) = \bigcup_{\delta \in \Delta} S_{s,\delta}$.

We now define these $S_{s,\delta}$.  Assume $s = (q,T)$.  If $\delta$ is a transition which does not go to state $q$, $S_{s,\delta} = \emptyset$.  Otherwise, $\delta = p \xrightarrow{a, (i, \bar{s})} q$ for appropriate $p$, $a$, $i$, and $\bar{s}$. In this case $S_{s, \delta} = \set{ T_{\rho, \bar{s}} \: : \: \rho \in R_{T,  \bar{s}}}$, where 
\begin{align*} R_{T,\bar{s}} = & \text{ the set of downward-paths, } n_0, n_1, \dots n_j, \text{ from the root, such that the labels of} \\
& \text{this path (excluding the root's label) give some non-empty prefix of }  q^k \text{ and} \\
& \text{if } s_k = \bot \text{ then } n_k  \text{has no children other than (possibly) } n_{k+1}
\end{align*} 
$T_{\rho, \bar{s}}$ is then defined as follows:
\begin{itemize}
\item If $\rho$ is a path of $i$ nodes (i.e. the path labels were $q^i$), then $T_{\rho, \bar{s}}$ is constructed by replacing the labels of the nodes in $\rho$ with the corresponding element of $\bar{s}$.  If such an element is $\bot$ then the node is deleted (note that in this case by choice of $\rho$ any child nodes must also be in $\rho$ and labelled with $\bot$, so are also deleted)\footnote{Note that the union of the $T_{\rho, \bar{s}}$ for these paths actually give $Pred((q,T))$.  The upward closure is obtained from the shorter paths in the next case.}
\item If $\rho$ is a path of $1 \leqslant j < i$ nodes, then $T_{\rho, \bar{S}}$ is constructed by replacing the labels of the nodes in $\rho$ with the corresponding element of $\bar{s}$, and adding nodes labelled with the remaining elements of $\bar{s}$ as a branch off the last node in $\rho$.  Any nodes labelled with $\bot$ are then deleted (again, due to the choice of $\rho$ and constraints on transitions this must result in obtaining a tree).
\end{itemize}

%% file: appendix-homca-proof.tex
We note that HOMCA and HOMCA' clearly have equal expressivity in the level-$1$ and level-$2$ cases.

In general, to simulate a HOMCA' with a HOMCA is straightforward: the only extra work that needs be done is prevent $store_i$ or $new_i$ operations happening when they are not allowed to in a HOMCA'.  This can be done by storing which of the current configuration's multisets are currently undefined in the states.

Simulating a HOMCA with a HOMCA' requires dealing with configurations $(q, m_1, \dots, m_k)$ where for some $i < j < k$ $m_j = \bot$ and $m_i \neq \bot \neq m_k$.  In particular, $m_1, \dots, m_{j-1}$ may then be split off into a new level-$j$ multiset.  
We describe how this can be dealt with in the level-$3$ case, and extensions of this method can be used to deal with higher levels.  The main idea will be, when in configuration $(q, m_1, m_2, m_3)$, $m_1 \neq \bot$ and a $store_2$ operation would occur, to create a new ``ghost" level-$2$ multiset, to copy $m_1$ across to ``under" this multiset.  This is done by the HOMCA' repeatedly performing the following operation: removing one letter from $m_1$, then storing $m_1$ in $m_2$, storing $m_2$ in $m_3$, loading the ``ghost" level-$2$ multiset and then the loading the copy of $m_1$, and adding the letter removed from $m_1$ to the copy of $m_1$.  At the end of this process, the original $m_1$ is marked as ``inactive", and will not be used again.  If anything was left in $m_1$, it will still be there at the end of the run, so the run will not be an accepting one.  When, later, a new level-$2$ multiset would be created, this can be done by simply removing the ``ghost"-marking from the level-$2$ multiset already created.

Formally, we add to the multiset alphabet, $\Gamma$, new letters of the form $(active,i)$, $(inactive, i)$, $(mf,i)$, $(mt,i)$, $(current,i)$, and $(ghost, i)$, where $i \in \set{1, 2, 3}$.  
The $mf$ and $mt$ tags will be used to keep track of which multisets we are ``moving from'' and ``moving to'' in the folding-and-unfolding process described above.
Whenever a level-$i$ multiset is made, it is immediately populated (using a series of $\epsilon$-transitions) to contain the level-$(i-1)$ multiset $\set{(active,i)}^{i-1}$.\footnote{Following notation from \cite{BjorklundB07}, we write $\set{a}^j$ to mean the level-$j$ multiset $\set{\dots \set{a} \dots}$.}  We can then, when a level-$i$ multiset is loaded, check for the ``active" tag by making the counter operations:
$$ load_{i-1} \cdot \dots \cdot load_1 \cdot  dec_{(active,i)} \cdot inc_{(active,i)} \cdot store_1 \cdot \dots \cdot store_{i-1}$$
We abbreviate this sequence of operations to $check_{active,i}$.  
Now, when a $store_2$ operation would happen with a non-$\bot$ level-$1$ multiset open, we make the following sequence of counter operations:
\begin{itemize}
\item First, the level-$1$ multiset is changed from containing $(active,1)$ to containing $(mf,1)$, and stored in the level-$2$ multiset (with operations $dec_{(active,1)} \cdot inc_{(mf,1)} \cdot store_1$)
\item The level-$2$ multiset is also marked with operations $ load_1 \cdot dec_{(active,2)} \cdot inc_{(mf,2)} \cdot store_1$
\item The ``ghost" level-$2$ multiset is then initialised with operations $store_2 \cdot new_2 \cdot new_1 \cdot inc_{(mt,2)} \cdot store_1$
\item The new copy of $m_1$ is initialised with operations $new_1 \cdot inc_{(mt,1)} $
\item Then we return to the original copy of $m_1$ with operations $store_1 \cdot store_2 \cdot load_2 \cdot check_{mf,2} \cdot load_1 \cdot check_{mf,1}$
\item We then non-deterministically run the following operations any number of times (with possibly different $\gamma \in \Gamma$ on each iteration):
\begin{align*} &dec_\gamma \, \cdot \, store_1 \, \cdot \, store_2 \, \cdot \,  load_2 \, \cdot \, check_{mt,2} \, \cdot \, load_1 \, \cdot \, check_{mt,1} \, \, \cdot \\ 
&inc_\gamma \, \cdot \,  store_1 \, \cdot \, store_2 \, \cdot \, load_2 \, \cdot \, check_{mf,2} \, \cdot \, load_1 \, \cdot \, check_{mf,1} 
\end{align*}
\item After this has run non-deterministically many times, we assume that all of $m_1$ has been copied, and change its marking to inactive, with the operations $dec_{(mf,1)} \cdot inc_{(inactive,1)} \cdot store_1$
\item We then return the marking of $m_2$ to normal and store it as required, with the operations $load_1 \cdot dec_{(mf,2)} \cdot inc_{(active,2)} \cdot store_1 \cdot store_2$
\item Finally, we correct the markings of the new multisets with the following operations: $load_2 \cdot load_1 \cdot dec_{(mt,2)} \cdot inc_{(ghost,2)} \cdot store_1 \cdot load_1 \cdot dec_{(mt,1)} \cdot inc_{(active,1)}$
\end{itemize}
When a $new_2$ operation would subsequently happen, the current level-$1$ multiset at the time would have its marking changed from $(active,1)$ to $(current,1)$.  This is then stored, and the level-$2$ multiset has its marking changed from $(ghost,2)$ to $(active,2)$.  Finally, a $load_1$ operation is made, and the loaded multiset has its marking changed from $(current,1)$ to $(active,1)$.  At the end of the run, the automaton is able to remove any $active$, $inactive$, $current$, and $ghost$ markings.

This construction guarantees that every element is moved across correctly, because if not the ignored element will still be present at the end of the run, violating hereditary emptiness.  In the weak case, this cannot be enforced, and this operation may behave ``lossily''.  However, this is not a problem: if it is possible for there to be a ``lossy'' run, then it is possible for there to be a non-lossy run.  If there is a lossy run which is accepting, then the corresponding non-lossy run must also be accepting, since increments of counters cannot prevent acceptance.

We have just given a basic example, but the same ideas can be used for the constructions necessary for dealing with $store_i$ operations with $i > 2$.  In this example we only needed to move across a single level-1 multiset, which we did by iterating over each of the elements in the multiset.  When moving a level-2 multiset, each level-1 multiset must be moved, by iterating over them, using the method described here for each one.

%% file: appendix-ndcma-homca.tex
For this proof we make use of some slight syntactic sugar for NDCMA, which does not change their power or other properties.  Instead of the transition function just giving a single state which the automaton moves to and which class memory function stores for the current data value and all of its ancestors, the transition function can specify the class memory function's newly stored value for the input data value and each of its ancestors individually.  We also add a level-0 data value that is the parent of every level-1 data value (and we require that only data values of level-1 and above can actually be read). Formally, the transition function is now $\delta = \bigcup_{1 \leq i \leq l} \delta_i$ where each $\delta_i$ is a function:
\[
\delta_i : Q \times \Sigma \times (\set{i} \times (Q_\bot)^{i+1}) \rightarrow \powerset (Q \times Q^{i+1})
\]
When the automaton is in configuration $(q,f)$, reads input letter $\twovec{a}{d}$, and takes transition $(q, a, (i,\threevec{s_0}{\vdots}{s_i})) \rightarrow (q', \threevec{t_0}{\vdots}{t_i})$, the automaton moves to configuration $(q', f')$ where $f' = f[pred^i(d) \mapsto t_0, pred^{i-1}(d) \mapsto t_1, \dots, d \mapsto t_i]$.  Notationally, we may write this transition as:
$$ p \xrightarrow{ a, (i, \threevec{s_0}{\vdots}{s_i})} q, \threevec{t_0}{\vdots}{t_i} $$
A level-$k$ NDCMA $\anglebra{Q, \Sigma, \delta, q_0, F_L, F_G}$ with this syntactic sugar can be simulated by a level-$k$ NDCMA with state-set $Q^{k+1}$, by using the $(j+1)^{th}$ component of the state to store the specified location for the level-$j$ ancestor of the read data value.

We prove the result by showing the two simulation separately.
\begin{proposition}
Given a strong (resp. weak) NDCMA $\calA$, a strong (resp. weak) HOMCA' $\calA'$ of the same level can be constructed such that $\mathcal{L}(\calA') = str( \mathcal{L}(\calA)) $
\end{proposition}
\proof{
Following notation from \cite{BjorklundB07}, for a letter of the multiset alphabet $a$ we write $\set{a}^1$ for a level-$1$ multiset containing just the letter $a$ (once), and $\set{a}^{n+1}$ for a level-$(n+1)$ multiset containing just the multiset $\set{a}^n$ (once).
Data values are then represented as follows:
\begin{itemize}
\item level-$k$ data values are stored as letters in level-$1$ multisets: a letter $(s,k)$ represents a level-$k$ data value remembered as being in state $s$.
\item level-$(k-1)$ data values are stored as level-$1$ multisets, containing a letter for each nested data value, and a letter $(s,k-1)$ (where $s$ is the state the data value is remembered as being in).
\item level-$(n-1)$ data values are stored as level-$(k- n +1)$ multisets, containing a level-$(k-n+1)$ multiset for each nested data value, and a multiset $\set{(s,n-1)}^{k-n}$.
\end{itemize}
Let $\calA = (Q, \Sigma, \Delta, q_0, F_L, F_G)$ be a level-$k$ NDCMA, we define the HOCMA $\calA'$ as follows:
\begin{itemize}
\item $\calA'$ is a level-$k$ HOCMA
\item The multiset alphabet is $Q \times \set{0, \dots, k}$
\item The states are $Q \uplus \set{q_F} \uplus \bigcup_{\delta \in \Delta} T_\delta$ (where $T_\delta$ is some set of states which will be used in simulating $\delta$-transitions).
\item The initial state is $q_0$
\item $q_F$ is the only accepting state
\item The transitions are given as follows:
\begin{itemize}
	\item For each transition $\delta = p \xrightarrow{ a, (i, \threevec{p_0}{\vdots}{p_i})} q, \threevec{q_0}{\vdots}{q_i}$ we have a sequence of transitions starting from $p$ and ending in $q$, going through the states $T_\delta$, which make counter operations as follows:
	\begin{itemize}
		\item We assume that at the start of these transitions, the configuration has multisets $(\bot, \dots, \bot, m_k)$.  (i.e. only the top-level multiset is defined).  (If $p$ is the initial state this may not be the case, but the way to amend the following construction when the 0th level data value hasn't been seen is straightforward).
		\item We make counter operations:
		 $$load_{k-1} \cdot load_{k-2} \cdot  \dots \cdot load_1 \cdot dec_{(p_0,0)} \cdot inc_{(q_0,0)} \cdot store_1 \cdot store_2 \cdot \dots \cdot store_{k-1}$$ to ensure that the level-0 data value is in the correct place, and to update it.
		\item We then $load_{k-1}$ to select a level-1 data value which will be used for this transition, unless $p_1 = \bot$ in which case $new_{k-1}$ is run.
		\item We then do operations $$load_{k-2} \cdot \dots \cdot load_1 \cdot dec_{(p_1,1) \cdot} inc_{(q_1,1)} \cdot store_1 \cdot store_2 \cdot \dots \cdot store_{k-2}$$ to correctly check and update the level-1 data value.  (Unless $p_1 = \bot$, in which case the $load_i$ operations are replaced with $new_i$, and the $dec$ operation is omitted.)
		\item This process is repeated until we have checked the level-$i$ data value by doing $dec_{(p_i,i)} inc_{(q_i, i)}$.  Then all of the opened multisets are stored, up to $store_{k-1}$.
	\end{itemize}
	All but the first of these transitions are labelled with $\epsilon$ (and the first is labelled with $a$)
	\item From states in $F_G$ we have an $\epsilon$ transition to the state $q_F$
	\item From state $q_F$ there are transitions to make any load or store operation, and to make any $dec_{(s,i)}$ where $s \in F_L$.
\end{itemize}
\end{itemize}

It is a straightforward induction to show that $\calA'$ simulates $\calA$.  If $\calA$ is weak, then it is clear that from $q_F$ any configuration can be reduced to a hereditarily empty one, hence a weak HOCMA' is enough.
}

\begin{proposition}
Given a strong (resp. weak) HOCMA' $\calA'$, a strong (resp. weak) NDCMA $\calA$ of the same level can be constructed such that $str( \mathcal{L}(\calA)) = \mathcal{L}(\calA')$
\end{proposition}
\proof{
We note that by definition of HOCMA', the only reachable configurations are of the form $(q, m_1, \dots , m_n)$ such that for some $i \leq n$ we have that for all $j < i$, $m_j = \bot$, and for all $j' \geq i$, $m_{j'} \neq \bot$. (Hence the first counter operation used must be $new_n$.)  In the following construction we will use the remembered location of the level-0 data value to track the value of this $i$.

Given a level-$k$ HOMCA $\calA' = \anglebra{Q', \Sigma, \Delta', q_0', F'}$ over a multiset alphabet $A$, we construct a level-$k$ NDCMA $\calA = \anglebra{Q, \Sigma, \Delta, q_0, F_L, F_G}$ that recognises the same string-projection as $\calA'$ as follows:
\begin{itemize}
\item $Q = Q' \uplus A \uplus \set{(1), \dots, (k)} \uplus \set{ \bullet, \circ} \uplus \set{q_{dead}}$.  Here the states $(i)$ will be used to keep track of the smallest $i$ such that $m_i \neq \bot$ in the current configuration, and the state $\bullet$ will be used to keep track of which data values are representing multisets currently ``active" (i.e. equal to one of the $m_i$), while the state $\circ$ will store the ``inactive" multisets.
\item $q_0 = q_0'$
\item $F_G = F$
\item $F_L = \set{q_{dead}} \cup \set{(1), \dots, (k)} \cup \set{\bullet, \circ}$
\item If $p \xrightarrow{a, op} q \in \Delta'$ then we have a transition $p \xrightarrow{a, (i, \threevec{p_0}{\vdots}{p_i})} q, \threevec{q_0}{\vdots}{q_i}$ where $i$, $\threevec{p_0}{\vdots}{p_i}$ and $\threevec{q_0}{\vdots}{q_i}$ are such that:
\begin{itemize}
	\item If $op = new_k$ then $i=0$ and $p_0 = \bot$ and $q_0 = (k)$
	
	\item If $op = new_m$ when $m \neq k$ then $i = k - m$ and:
	\begin{itemize}
		\item $p_0 = (m+1)$ and $q_0 = (m)$
		\item for $0 < j < i$ $p_j = \bullet = q_j$
		\item $p_i = \bot$ and $q_i = \bullet$
	\end{itemize}	 
	
	\item If $op = inc_b$ ($b \in A$) then $i=k$ and 
	\begin{itemize}
		\item $p_0 = (1) = q_0$ 
		\item for $0 < j < k$ $p_j = \bullet = q_j$
		\item $p_k = \bot$ and $q_k = b$
	\end{itemize}
	
	\item If $op = dec_b$ ($b \in A$) then $i=k$ and 
	\begin{itemize}
		\item $p_0 = (1) = q_0$ 
		\item for $0 < j < k$ $p_j = \bullet = q_j$
		\item $p_k = b$ and $q_k = q_{dead}$
	\end{itemize}
	
	\item If $op = load_m$ ($m < k$) then $i = k - m$ and:
	\begin{itemize}
		\item $p_0 = (m+1)$ and $q_0 = (m)$
		\item for $0 < j < i$ $p_j = \bullet = q_j$
		\item $p_i = \circ$ and $q_i = \bullet$
	\end{itemize}
	
	\item If $op = store_m$ ($m < k$) then $i = k - m$ and:
	\begin{itemize}
		\item $p_0 = (m)$ and $q_0 = (m+1)$
		\item for $0 < j < i$ $p_j = \bullet = q_j$
		\item $p_i = \bullet$ and $q_i = \circ$
	\end{itemize}
\end{itemize}
\end{itemize}
It is again straightforward to show this simulates the run of $\calA'$.  The acceptance condition is checked by $F_L$ not including $A$: the only way for the multiset to not be hereditarily empty is for an increment never to have a corresponding decrement, and this (and only this) can leave a data value in $A$.  If we have a weak HOCMA', then we do not need to check this and can let $F_L$ be the whole set of states.
}

This completes the proof of theorem \ref{thm:ndcma-homca-equivalence}.  We note that our earlier theorem \ref{thm:ndcma-emptiness} is now a corollary of emptiness of HOMCA being undecidable (shown in \cite{BjorklundB07}).